\documentclass[a4paper,10pt]{article}
\usepackage[all]{xy}
\usepackage{setspace}
\usepackage{amsfonts}
\usepackage{tikz-cd}
\newtheorem{theorem}{Theorem}

\newtheorem{proof}{Proof}
\newcommand{\qed}{\hfill \ensuremath{\Box}}
\usepackage{graphicx}% http://ctan.org/pkg/graphicx

\makeatletter
\let\@fnsymbol\@arabic
\makeatother

\usepackage{hyperref}
\usepackage{graphicx}
\usepackage{bookmark}
\usepackage{epsf}
\usepackage{epsfig}
\usepackage{epstopdf}
\usepackage{colordvi}
\usepackage{color}
\usepackage{amscd}
\usepackage{amsmath}
\usepackage{amsfonts}
\usepackage{amssymb}
\usepackage{mathrsfs}
\usepackage{geometry}
\usepackage{times}                 % mandatory
\usepackage{amssymb}      % for \mathbb{Z}
\usepackage{amsmath}
\usepackage{subfigure}
\usepackage{bm}% bold math

\begin{document}
\title{Spread of Crime Dynamics: A mathematical  approach}%
\author{Michael Aguadze$, ^{1}$ \vspace*{0.2cm} Ana Vivas$,^{2}$ \vspace*{0.2cm} Sujan Pant$,^{2}$ \vspace*{0.2cm} Kubilay Dagtoros$^{2}$\\
$^{1}${Department of Electrical Engineering, Norfolk State University, Norfolk, United States},\\
$^{2}${Department of  Mathematics, Norfolk State University, Norfolk, United States}
}
\maketitle

\begin{abstract}
In this work, the spread of crime dynamics in the US is analyzed from a mathematical scope, an epidemiological model is established, including five compartments: Susceptible ($S$), Latent 1 ($E_1$), Latent 2 ($E_2)$, Incarcerated ($I$), and Recovered ($R$). A system of differential equations is used to model the spread of crime. A result to show the positivity of the solutions for the system is included. The basic reproduction number and the stability for the disease-free equilibrium results are calculated following epidemiological theories. Numerical simulations are performed with US parameter values. Understanding the dynamics of the spread of crime helps to determine what factors may work best together to reduce violent crime.
\end{abstract}

%%%%%%%%%%%%%%%%%%%%%%%%%%%%%%%%%%%%%%%%%%%%%%%%

\section{Introduction}

The United States of America has one of the highest per capita incarceration rates in the world \cite{FBI_report}. In fact, every single US state has higher incarceration rates than most nations on Earth . The number of incarcerated individuals and all the entities involved in this process cost billions of dollars to the US taxpayers. Criminal activity has been considered a contagious phenomenon by several authors [refs]. In this work, an epidemiological model is presented under the assumption that an individual can begin criminal activity induced by the influence of their peers or by an intrinsic desire to commit criminal activity without being induced by others.\\ 
Criminality is a social phenomenon that can be spread within social communities that share a common demographic identity that includes race, ethnicity, economic opportunity, education, and political socialization, many authors studied this phenomena from different perspectives \cite{Clinard}, \cite{Becker},\cite{Glueck},\cite{Santoja},\cite{Ehrlich}.  Further, relevant literature indicates that criminality and recidivism can be largely attributed to structural social disparities embedded in the legal, political, and economic institutions \cite{Akers},\cite{Crane}. Additionally, some studies indicated that criminal tendency is more prominent when an individual has experienced a childhood trauma, a study conducted for incarcerated women depicts this relation in \cite{Lehrer}. Previous work had been done using compartmental models to study the dynamics of crime \cite{McMillon}, \cite{Zhao_Feng_CCC},\cite{Gordon} This work aims to provide an understanding of this social science phenomenon through a mathematical lens. A compartmentalized modeling method is used to understand the dynamics of at-risk populations, \cite{Brauer},\cite{Banks_CCC},\cite{Hethcote},\cite{Kretzschmar},\cite{Lawson_Marion}.\\  
The model assumes that the total population $N$ is divided into five compartments, which include:   S (Susceptible individuals with no criminal behavior), $E_1$ (Latent 1 individuals with criminal behavior, who never entered the legal system and never being incarcerated), $E_2$ (Latent 2 repeat offenders, who were incarcerated, released and committed a crime again), I (Incarcerated individuals), and R (Released individuals).  In the analysis, we evaluate the basic reproduction number \cite{CCC_ZFeng},\cite{Diekman},\cite{VanWat}, and calculate the disease-free equilibrium, the endemic equilibrium, and stability. 
Additionally, simulations are included using data for the USA. Simulations were performed with data from the at-risk population, which primarily are made up of racial and ethnic minorities experiencing socio-economic disparities. The results reinforce the findings that crime and incarceration rates are associated with structural inequalities stemming from racial prejudice towards minority communities. Hence, this analysis has policy implications for understanding the spread and dynamics of violent crime.

The second section includes the description of the compartments and the parameters used to model the dynamics for the spread of crime. For modeling purposes, assumptions corresponding to the compartments are included {ref}, as well as a description of some human behaviors for individuals who had been incarcerated and reasons why, after they had been released from jail or prison, two-thirds of them tend to commit crimes again [ref]. 

In the third section, the mathematical analysis of the model including results for the positivity of the solutions and the stability of the free-disease equilibrium,  the basic reproduction number (rate of secondary infection after being exposed to a criminal individual) evaluated based on the next-generation matrix method \cite{VanWat}. Numerical simulations for the reproduction number are included in this section.

In the fourth section, numerical simulations are included for different scenarios, varying the most relevant parameters in the model. Identifying the parameters in the model to find a way to increase the recovery flow of first-time offenders back to law-abiding citizens is a priority of this work.

%%%%%%%%%%%%%%%%%%%%%%%%%%%%%%%%%%%%%model%%%%%%%%%%%%%%%%%%%%%%%%%%

\section{Modeling the Spread of Crime}

The spread of crime dynamics is considered a contagious disease, due to social interactions in the community, several factors can influence an individual to commit a crime. Assuming that the total population is divided into five groups, Susceptible ($S$) (not infected potentially at risk), 	Exposed 1  ($E_1$) (has the disease for the first time, in this case, committing a crime, but neither caught nor found guilty).Infected ($I$) (in this case incarcerated as they are showing as being "infected" whereas others exposed may be committing crimes but did not enter the judicial system, based on crimes reported). Recovered individuals ($R$) (in this case it means released from incarceration). Exposed 2 ($E_2$) (has gone back to committing crime) and Exposed 2 could flow into infected (incarcerated). Along with the groups mentioned above, the model allows for population flow from Susceptible (or law-abiding) to Exposed 1, or (committing crimes), some individuals can commit a crime without having contact with other criminally active individuals, Also, someone who has recovered may transition to law-abiding, and avoiding the risk of Exposed 2 (Recidivism)- Recidivism refers to the behavior as a result in the rearrest or re conviction or a return to prison within 3 years after release \cite{NCJ2018}. 

Most inmates are not prepared for a successful return to society. Risk factors of trauma and the cycle of criminal activity have strong connections to impaired decision-making, damaged social relationships, addiction, and compromised physical well-being (Cotter et al., 2016). To reduce recidivism rates, improve the lives of those who have experienced trauma, and prevent the future cycle of victims, an innovative approach to treatment and care of those in correctional settings is crucial.. Table 1 includes a description of the parameters, and Figure 1 shows the Flow Diagram between compartments.

The following flow chart represents the dynamics for the transition of individuals between compartments: 

 \[
\xymatrix{
& & & & & &&\\
\ar[r]_{\Lambda} &
\boxed{S}\ar[u]_{d} \ar@/^/[rr]^{\beta_{10}} \ar@/_/[rr]_{\beta_{11}}
&&\boxed{E_1} 
\ar@/^45pt/[rrrr]^{\gamma_1}
\ar[u]_{d} \ar[rr]_{\theta_1} &&\boxed{I} 
\ar@/_/[u]_{d} \ar@/^/[u]^{\delta}
 \ar[rr]^{\gamma_2} & & 
 \boxed{R} \ar@/_/[dll]_{\theta_3}\\
 &&&&&\boxed{E_2} \ar@/_/[urr]_{\gamma_3} \ar[u]^{\theta_2} \ar[d]^{d}&&\\
 & & & & &&
}\]

\begin{table*}
\caption{The description of parameters}
\begin{center}\label {Table 1}
\begin{tabular}{|| p{0.5in} | p{3 in} | p{0.8in} ||}\hline \hline
\small
Parameters & Description & value(default) \\ \hline \hline
$\Lambda$ & Population birth rate & 0.012 \cite{Martin}\\ \hline
$\beta_{10} $ &Effective rate into criminal activity without having contact with other criminally
active people & 0.18 \cite{BJSDead}\\\hline
$\beta_{11} $ & Effective contact rate into criminal activity as a result of having contact with other criminally active people & 0.24 \cite{FBI_report}\\ \hline
$ \theta_1 $ & Incarceration rate (police, courts, correctional systems) & 0.0035 \cite{Martin}\\ \hline
$ \theta_2 $ & Re-Incarceration rate (police, courts, correctional systems) & 0.0044 \cite{NCJ2018}\\ \hline
$ \theta_3 $ & Recidivist rate (repeated offenders)& 0.6666 \cite{NCJ2018} \\ \hline
$\gamma_1 $ & Recovery rate from criminal activity not related to the experience of incarceration &  0.9933 \cite{NCJ2018}\\ \hline
$\gamma_2 $ & Recovery rate from criminal activity related to
the experience of incarceration & 0.794 \cite{NCJ2018}\\ \hline
$ \gamma_3 $ & Recovery rate for recurrent criminals related with the experience  incarceration. & 0.3334 \cite{NCJ2018}\\ \hline
$\delta$ & Mortality rate associated with incarceration & 0.0033 \cite{BJSDead}\\ \hline
$d$ & Population death rate & 0.0132 \cite{Ahmad}\\
\hline
 \end{tabular}
 \end{center}
\end{table*} 
 
The dynamics for the spread of crime can be modeled by the following system of ordinary differential equations

\begin{align}\label{ODE_SYS}
\frac{dS}{dt} &=\Lambda-\beta_{11}\frac{S(E_1+E_2)}{N}-(\beta_{10}+d)S\\
\frac{dE_1}{dt} &= \beta_{11}\frac{S(E_1+E_2)}{N}+\beta_{10}S-(d+\theta_1+\gamma_1)E_1\\
\frac{dE_2}{dt} &= \theta_3 R - (\theta_2+\gamma_3+d) E_2\\
\frac{dI}{dt} &=  \theta_1 E_1 + \theta_2 E_2 - (\gamma_2+\delta+d)I\\
\frac{dR}{dt} &= \gamma_1E_1+\gamma_2I+\gamma_3E_2-(\theta_3+d) R
\end{align}
where 
$$N = S + E_1 + E_2 + I + R$$

%%%%%%%%%%%%%%%%%%%%%%%%%%%%%%%%%%%%%%%%%%%%%%%%%%%%%%%%

\section{Mathematical Analysis of the Model}
This section includes a result to show the positiveness and long-term behavior for the Solutions of System \ref{ODE_SYS}.

\begin{theorem}
\label{theoremone}
If each compartment is non-negative at $t=0$, then each compartment is non-negative for time $t>0$.
Moreover,
$$\lim_{t\rightarrow \infty} N(t)\leq \frac{\Lambda}{d}.$$
\end{theorem}

\begin{proof}
Assume that $T$ is the maximum time for the epidemic. That is, 
$$T=\sup\left\{S>0,E_1\geq 0,E_{2}\geq 0, I\geq 0, R\geq 0\right\}\in [0,t].$$

Therefore for $T>0$, from $E_1$-equation of System \ref{ODE_SYS}, 

$$\frac{dE_1}{dt}\geq -(\theta_1+\gamma_1+d)E_1,$$

from which it holds

$$E_1(T)\geq E_1(0)\exp\left\{-(\theta_1+\gamma_1+d)t\right\},$$

hence, $E_1(T)\geq 0$ for all $T>0$.

From $E_2$-equation of System \ref{ODE_SYS}, 

$$\frac{dE_2}{dt}\geq -(\theta_2+\gamma_3+d)E_2$$
from which it hold

d$$E_2(T)\geq E_2(0)\exp\left\{-(\theta_2+\gamma_3+d)t\right\}$$

hence, $E_2(T)\geq 0$ for all $T>0$.

From $S$-equation of System \ref{ODE_SYS}, 

$$\frac{dS}{dt}\geq -\left\{\beta_{11}\frac{S(E_1+E_2)}{N}+\beta_{10}\right\}S$$

from which it hold
$$S(T)\geq S(0)\exp \left\{\int_0^t-\left\{\beta_{11}\frac{S(E_1+E_2)}{N}+\beta_{10}\right\}ds\right\}$$
hence, $S(T)\geq 0$ for all $T>0$.

The positiveness of the remaining compartments can be shown in a similar way.

The long-term behavior for the total population $N$ can be analyzed as follows: Adding the left hand sides from \ref{ODE_SYS}, it holds 
   $$\frac{dN}{dt}=\Lambda-\delta I- d N,$$
then    
    $$\frac{dN}{dt} \leq \Lambda-d N,$$
from which it holds
$$\frac{dN}{dt}+d N\leq \Lambda,$$
then 
$$N(t)\leq \frac{\Lambda}{d}+\left(N_0-\frac{\Lambda}{d}\right)\exp(-dt).$$

Since $(N_0-{\Lambda/d})$ is a constant and 
$d>0$, 
\begin{equation*}
\frac{\Lambda}{d}+\left(N_0-\frac{\Lambda}{d}\right)\exp(-d t)\rightarrow \frac{\Lambda}{d}\text{ as }t\to\infty.\end{equation*}
then $\lim_{t\rightarrow \infty} N(t)\leq \frac{\Lambda}{d}$ as desired.

The feasible region $D$, for System \ref{ODE_SYS} is therefore 

$$D=\left\{(S,E_1,E_2,I,R)\in \mathbb{R}^{5}_+\quad \vert \quad N\leq \frac{\Lambda}{d}\right\}.$$
$\qed$
\end{proof}

\subsection{The Basic Reproduction Number}

Several factors are in consideration when an individual is prompt to commit a crime, psychological factors as well as sociological factors play an important role when modeling crime spread. In this study the assumption that crime can be spread through social interactions, led to the interpretation of the basic reproduction number as the number of secondary infectious \cite{CCC_ZFeng},\cite{Diekman}, meaning the number of new criminals that surge after an effective contact between a susceptible individual and a criminal individual. To evaluate $\mathscr{R}_0$, the next generation matrix method \cite{VanWat} is used, System \ref{ODE_SYS} is rearranged for simplicity, notice that the new infectious are the individuals in the compartments $E_1$. The disease-free  equilibrium is denoted by $x_0=(E_1,E_2,S,I,R)=(0,0,N,0,0)$.
Following the notation of \cite{VanWat}, matrices $\mathscr {F}$, $\mathscr{V}$ are

\[\mathscr{F}= \left[ \begin {array}{cccccccc}  
 \beta_{11}\frac{S(E_1+E_2)}{(N)} \\
  \\
   0 \\
  \\
 0 \\
 \\
 0 \\
 \\
 0 
\end {array} \right], \begin{array}{l}
\end{array}
\mathscr{V}=\left[ \begin {array}{cccccccc}  
  -\beta_{10}S  + (\theta_1+\gamma_1+d)E_1\\
   \\
   -\theta_3 R + (\theta_2 + \gamma_3+d)E_2\\ 
   \\
    - \Lambda + \beta_{11} \frac{S(E_1+E_2)}{N} + (\beta_{10}+d)S\\
   \\

 - \theta_1 E_1 -\theta_2E_2+(\gamma_2+\delta+d)I \\
 \\
 - \gamma_1E_1 - \gamma_2 I-\gamma_3E_2+ + (\theta_3  + d) R 
\end{array} \right], \begin{array}{l}
\end{array}\]\\

According to the next generation method in \cite{VanWat},
matrices $F$ and $V$ defined by $F=D \mathscr{F}(x_0)$ and $V=D \mathscr{V}(x_0)$ are given by

\[F= \left[ \begin{array}{cccccccc}  
 
\beta_{11} & \beta_{11} & 0 & 0 & 0\\
 \\
0 & 0 & 0 & 0 & 0\\
\\
0 & 0 & 0 & 0 & 0\\
\\
0 & 0 & 0 & 0 & 0\\
\\
0 & 0 & 0 & 0 & 0\\
\\
\end{array} \right], \begin{array}{l}
\end{array}
V=\left[ \begin{array}{cccccccc}  
\\
\theta_1+\gamma_1 +d & 0 & -\beta_{10} & 0 & 0\\
\\
0 & \theta_2+\gamma_3+d & 0 & 0 & 0\\
\\
\beta_{11} & \beta_{11} & \beta_{10}+d & 0 & 0\\
\\
-\theta_{1} & -\theta_{2} & 0 &\gamma_2+\delta+d & 0 \\
\\
-\gamma_1 & -\gamma_3  & 0 & -\gamma_2 & \theta_3+d \\
\\
\end{array} \right], 
\begin{array}{l}
\end{array}\]\\

\[\]

The basic reproduction number is the spectrum value of $FV^{-1}$ and depends on some parameters of the model

$$\mathscr{R}_0=\frac{\beta_{11}(\beta_{10}+d)}{\beta_{10}\beta_{11}+(\theta_1+\gamma_1+d)(\beta_{10}+d)}$$

The importance of calculating the basic reproduction number is rooted in the epidemiological fact that if $\mathscr{R}_0<1$, an epidemic can be controlled, otherwise, the epidemic can become a pandemic. Creating measures to control the values that affect the most $R_0$ is the ultimate goal of this work. The following graph shows the variation of $\mathscr{R}$ with respect to some particular parameters.

\begin{figure}[ht!]
\centering
    \includegraphics[viewport=60 10 800 600,clip=true,scale=.4]{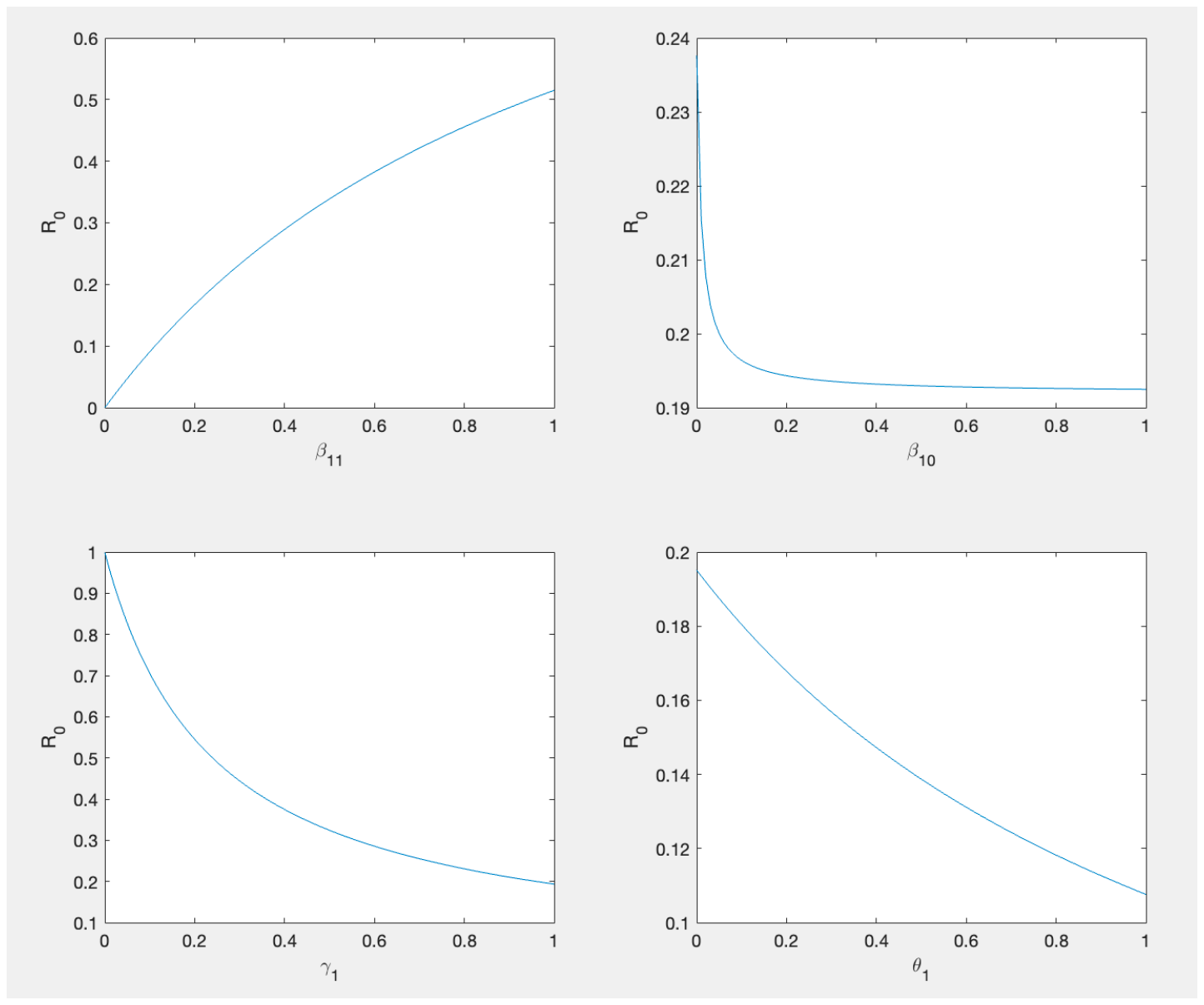}
    \vspace{-9 pt}
    \caption{$\mathscr{R}_{0}$ for $0\leq \beta_{11}\leq 1, 0\leq \beta_{10}\leq 1,0\leq \gamma_1\leq 1, 0\leq \theta_1 \leq 1$}
    \label{fig1}
    \vspace{-9 pt}
\end{figure}

In Figure 1, notice that $\mathscr{R}_{0}$ is increasing when $\beta_{11}$ is increasing, meaning that if susceptible individuals have effective contact with active criminal individuals then the number of secondary criminal individuals increases. Additionally, if the rate of incarceration $\theta_1$ increases notice that $\mathscr{R}_{0}$ decreases, which is in agreement with the fact that criminal individuals who are in jail cannot have contact with susceptible individuals and therefore the number of secondary criminals will be reduced. 

%%%%%%%%%%%%%%%%%%%%%%%%%%%%%%%%%%%%%%%%%%%%%%%%

\subsection{Stability for the disease-free equilibrium}

The local stability for the disease-free equilibrium $x_0$, can be shown by using the Jacobian for the System $\ref{ODE_SYS}$evaluated at $x_0$.The following result presents the eigenvalues for the Jacobian of System $\ref{ODE_SYS}$, and shows that under particular conditions ovet the parameters, the real part of the eigenvalues are negative.

\begin{theorem}
\label{Thorem 2}
The disease-free equilibrium $x_0$ for System \ref{ODE_SYS} is stable.
\end{theorem}
\begin{proof}
The Jacobian for System $\ref{ODE_SYS}$ at $x_0$, is given by

\[\]
\[J(x_0)=\left[ \begin{array}{cccccccc}  
\\
\beta_{11}-(d+\theta_1+\gamma_1) & \beta_{11} & \beta_{10} & 0 & 0 \\
\\
0 & -(\theta_2+\gamma_3+d) & 0 & 0 & 0\\
\\
-\beta_{11} & -\beta_{11} & -(\beta_{10}+d) & 0 & 0\\
\\
\theta_1 & \theta_2 & 0 & -(\gamma_2+\delta+d) & \theta_3 \\
\\
\gamma_1 & \gamma_3 & 0 & \gamma_{2} & -(\theta_3+d)\\
  \\
\end{array} \right], 
\begin{array}{l}
\end{array}\]\\
\[\]

The real part for the eigenvalues of $J(x_0)$ are given by
\begin{align*}    
\lambda_1&=-(\theta_2+\gamma_3+d)\\
\lambda_2&=-\left\{\frac{d}{2}+\frac{\gamma_2+\delta+d}{2}+\frac{\theta_3}{2}\right\}\\
\lambda_3&=-\left\{\frac{d}{2}+\frac{\gamma_2+\delta+d}{2}+\frac{\theta_3}{2}\right\}\\
\lambda_4&=\frac{\beta_{11}-(\beta_{10}+\theta_1+\gamma_1+2d)}{2}\\
\lambda_5&=\frac{\beta_{11}-(\beta_{10}+\theta_1+\gamma_1+2d)}{2}
\end{align*}
All the eigenvalues has a real negative part under the condition that $$\beta_{11}<\beta_{10}+\theta_1+\gamma_1+2d,$$

then the disease-free equilibrium is asymptotically stable, meaning that when $t$ tends to infinite the solutions approach to $x_0$. 
$\qed$
\end{proof}

The next section will include numerical simulations for System \ref{ODE_SYS}, using parameters from Table 1, also some simulations will vary some parameter values related to the basic reproduction number.  

%%%%%%%%%%%%%%%%%%%%%%%%%%%%%%%%%%%%%%%%%%%%%%%%%%%%
\section{Numerical Simulations}

Numerical simulations were performed using the most recent parameter values for the United States. Table 1 includes the references from where the parameters were obtained.

Figure 2 shows the dynamics of all the compartments for the model; notice that when time increases,the population in the latent $E_1$ and $E_2$ compartments do not vanished, meaning that there will be always a number of criminal active individuals and from those who serve a time in jail/prison and have being release, they are back to be criminally active.

\begin{figure}[ht!]
\centering
    \includegraphics[viewport=10 10 800 600,clip=true,scale=.5]{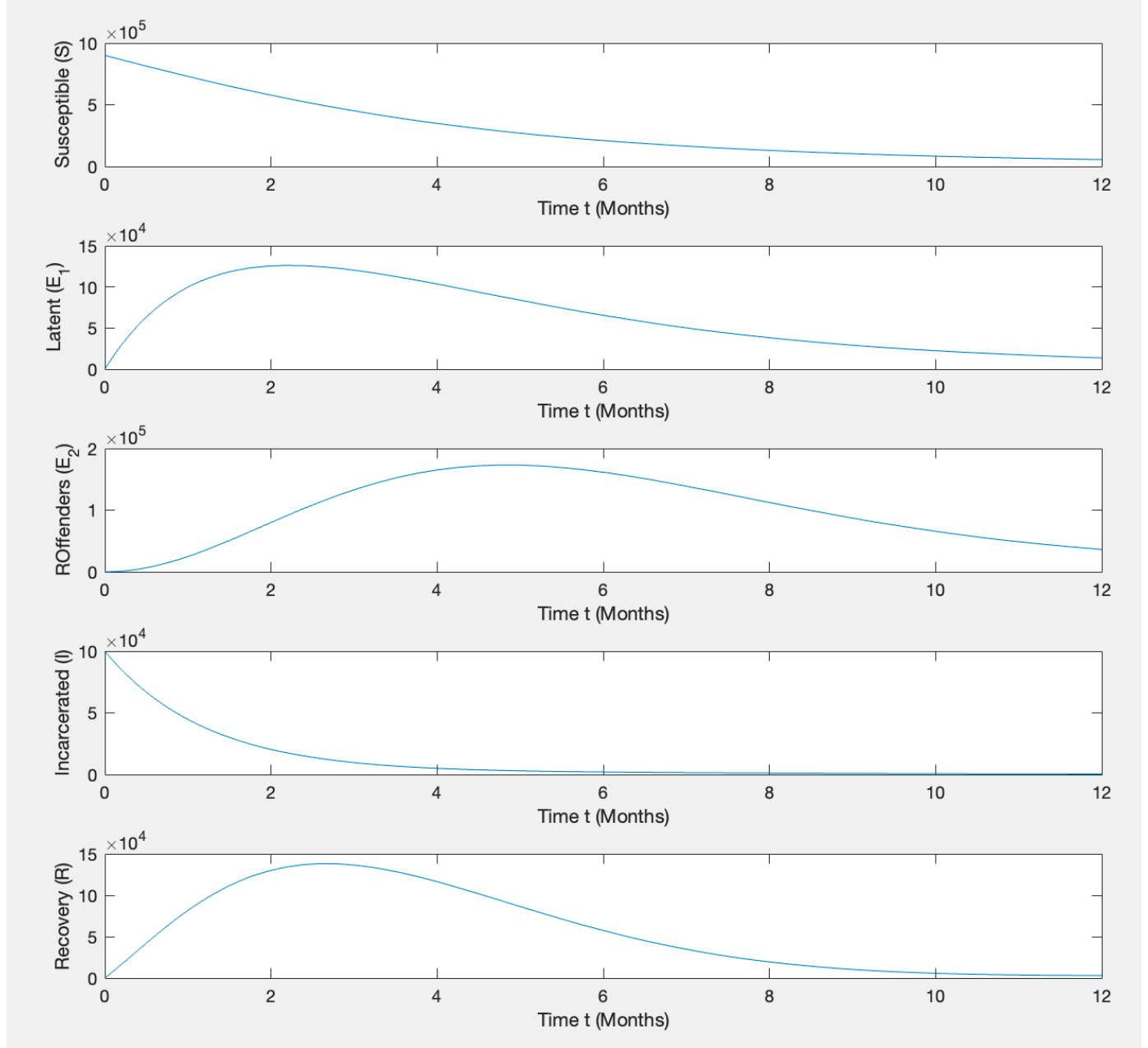}
    \vspace{-8 pt}
    \caption{$S,E_1, E_2,I, R$, for 12 months and parameters values from Table 1 }
    \label{fig2}
    \vspace{-9 pt}
\end{figure}

One of the motivations of this work is to analyze how the incarcerated population can be decreased. Figure 3 shows the behavior of the incarcerated population when varying $\theta_1$ (incarceration rate) and $\theta_2$ re-incarceration rate. Note that the number of incarcerated individuals increases, and has a maximum around month six at most of the levels when varying $\theta_2$ this is in agreement with the fact that if there are no programs to support released individuals from jail/prison to reincorporate them into society, they will be back to commit a crime. When varying $\theta_1$, the maximum values for incarcerated population are around month three levels when varying $\theta_1$, notice that if the incarceration rate $\theta_1$ is at a very low level, then the incarcerated population is decreasing during the time interval in consideration.

\begin{figure}[ht!]
\centering
    \includegraphics[viewport=10 140 900 480,clip=true,scale=.6]{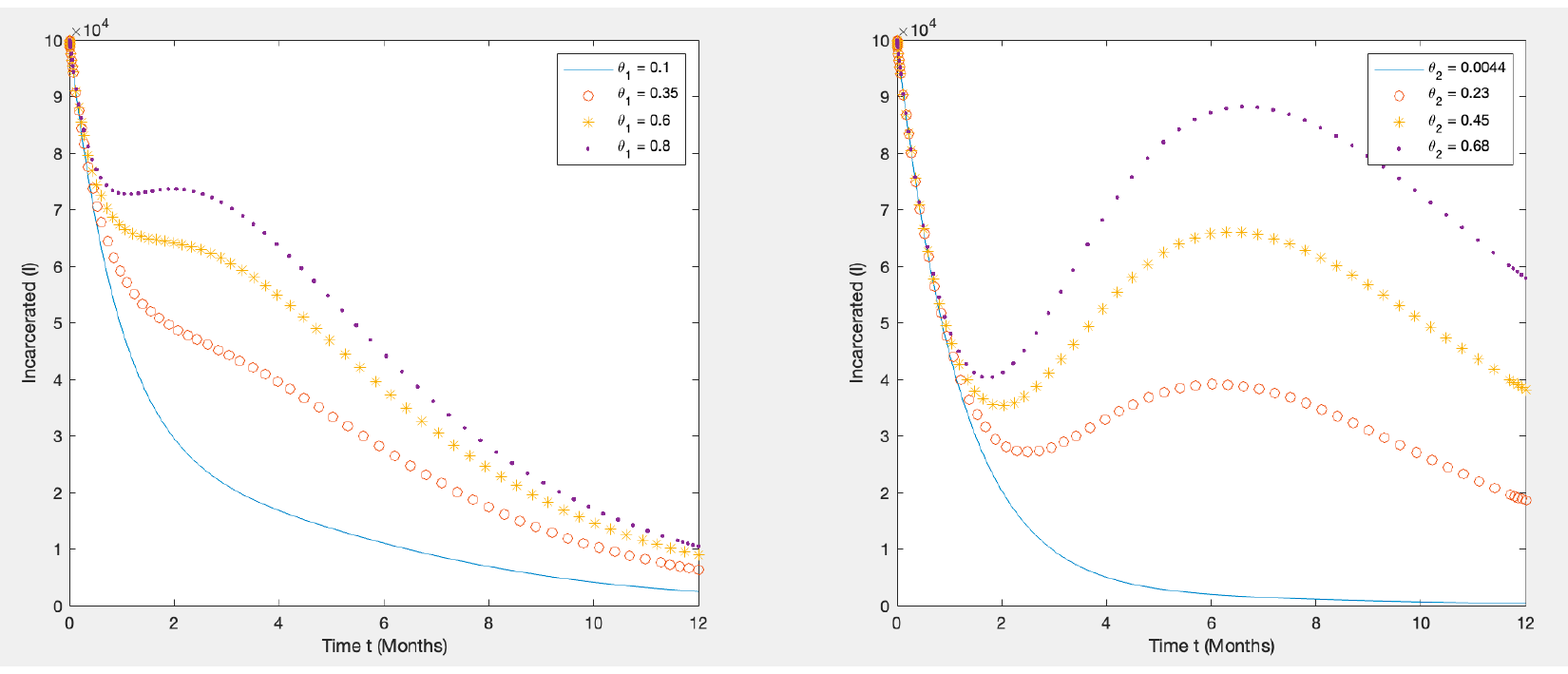}
    %\vspace{-8 pt}
    \caption{$I$ for $\theta_1=0.1, 0.35,0.6,0.8 $, and $\theta_2= 0.0044, 0.23, 0.45, 0.68$, other parameters are fixed.}
    %\label{fig3}
    %\vspace{-4 pt}
\end{figure}

A topic of interest when analyzing the dynamics of the spread of crime, is the influence of criminally active individuals with non-criminal individuals, the most vulnerable population is the teenagers within minorities, one important parameter to analyze is the effective contact rate $\beta_{11}$, in Figure 4 it can be observed that if $\beta_{11}$ increases then the number of individuals in the latent compartment $E_1$ increases as well, with the maximum values around the second and third months after the contact. Another influencing parameter for the latent population is $\theta_1$, in Figure 4, observe that for high values of $\theta_1$ the population in this compartment reduces significantly. 

\begin{figure}[ht!]
\centering
    \includegraphics[viewport=10 140 900 480,clip=true,scale=.6]{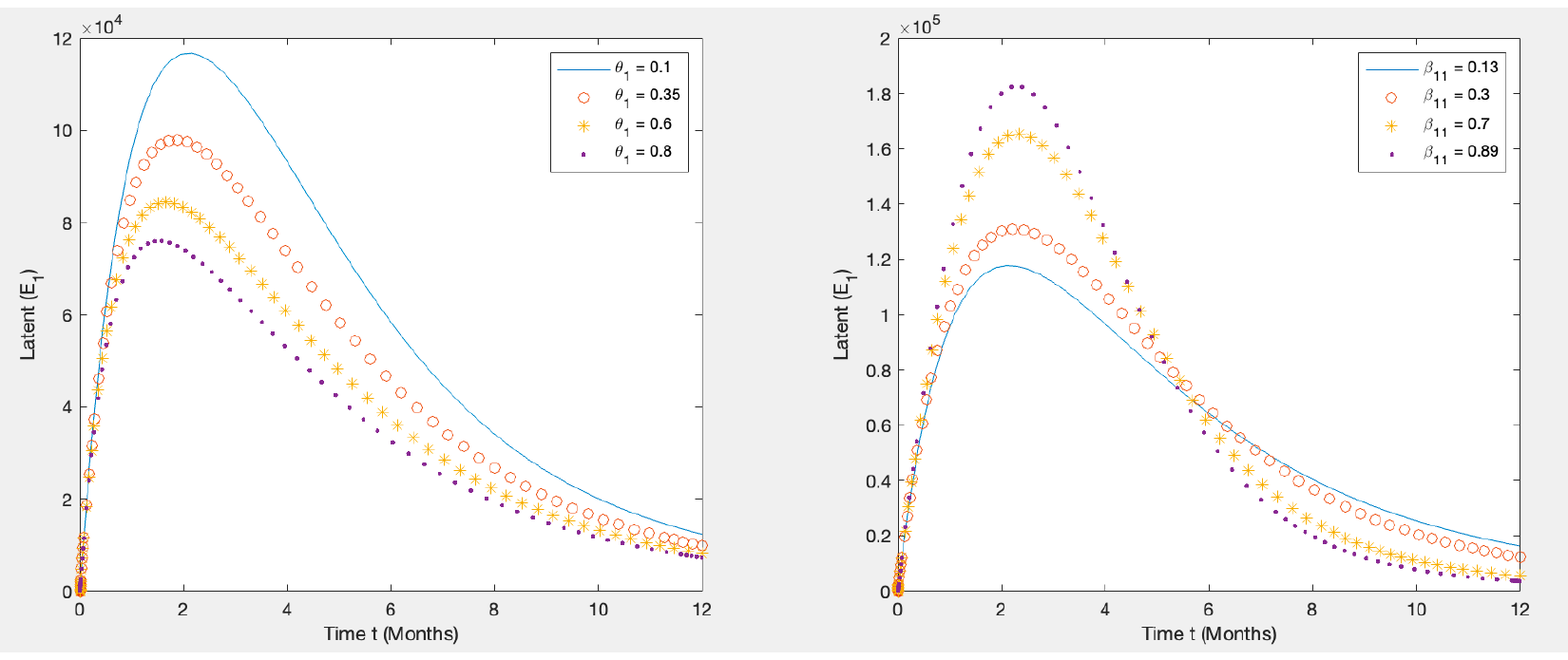}
    \vspace{-11 pt}
    \caption{$E_2$ for $\theta_1=0.1, 0.35,0.6,0.8 $, and $\beta_{11}= 0.13, 0.30, 0.70, 0.89$, other parameters are fixed.}
    \label{fig4}
    \vspace{-9 pt}
\end{figure}

Another relevant topic of interest is the $E_2$ population, individuals who  after had been released from jail/prison enter into criminal activity again and return to jail/prison. Figure 5 shows that the higher the value of $\theta_2$, the less individuals remain in this compartment, but most importantly if the recovery rate $\gamma_3$ is increasing then the population in this compartment will decreases. 

\begin{figure}[ht!]
\centering
    \includegraphics[viewport=10 140 900 480,clip=true,scale=.6]{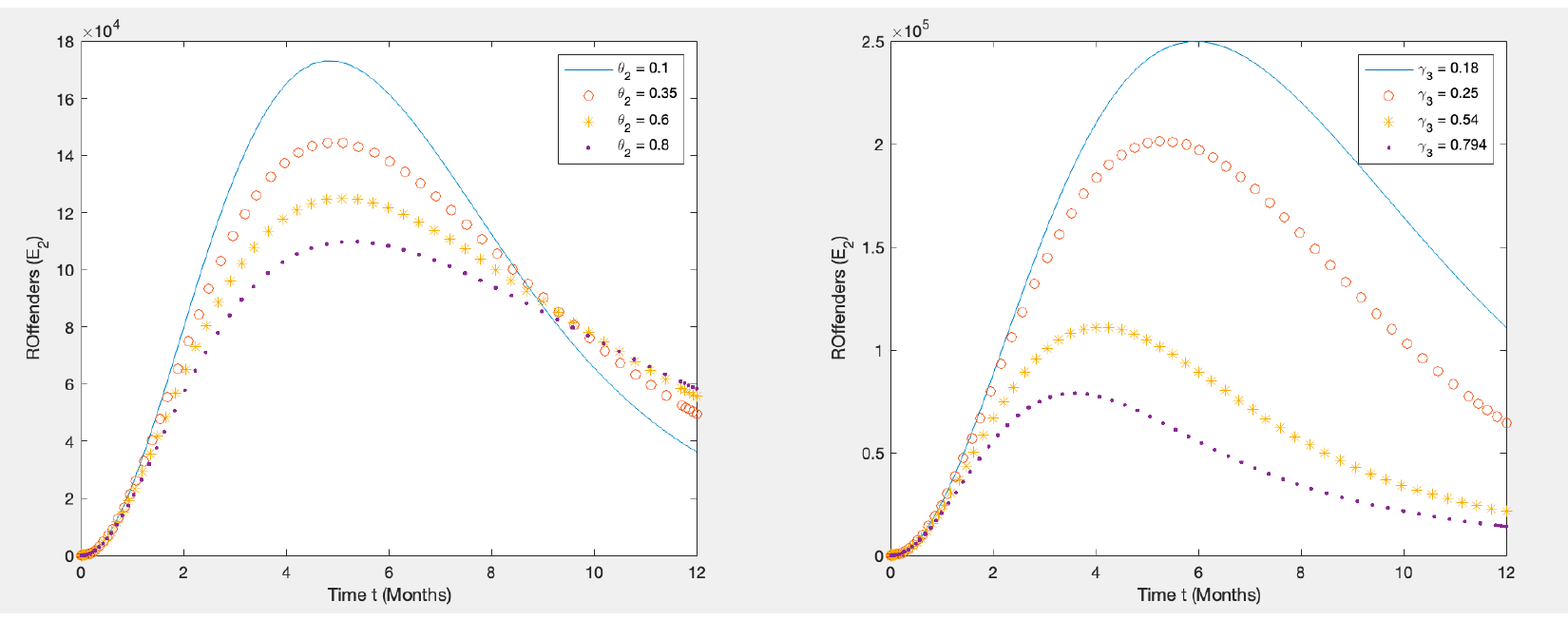}
    \vspace{-11 pt}
    \caption{$E_2$ for $\theta_2=0.1, 0.35,0.6,0.8 $, and $\gamma_3= 0.18, 0.25, 0.54, 0.794$, other parameters are fixed.}
    \label{fig5}
    \vspace{-9 pt}
\end{figure}

%%%%%%%%%%%%%%%%%%%%%%%%%%%%%%%%%%%

\section {Conclusions}

The spread of crime is a very complex problem that is affecting the United States population tremendously. To find feasible solutions to this problem multiple entities need to be involved. Evidence shows that the most vulnerable sectors of the population have a higher risk of being involved in criminal activities. Since incarceration and recidivism rates are higher for minorities, intervention programs should be offered for those sector of the populations, offering alternatives as education opportunities and extra-curricular activities, for individuals at high risk to commit criminal activities, and other measures from the local governments, will reduce the criminal activity among teens. After an individual is released from jail/prison, the local government should implement programs to support the individual and their families, facilitating the reinsertion of this individual to the society.

The number of secondary infected individuals depends on certain parameters, those parameters can be controlled using particular measures by the entities involved. For example, Reducing the contact rate between criminal-active individuals and non-criminal-active individuals will reduce the number of secondary infected individuals. 

Increasing the incarceration rate produces the effect of reducing the number of secondary infected individuals, which in epidemiology is one of the most important goals (to keep this number below one if possible). However, once an individual enter the criminal judicial system, to break this chain, it requires a tremendous effort from all the entities involved, it had been proved that recidivism is very high in communities with social and economical disadvantages. The judicial system in the US, should re-evaluate the types of crimes that deserve to send a young person into jail, because this is a prime age when an individual can be recovery from criminal activity with the appropriate support from the government and theirs families.

%%%%%%%%%%%%%%%%%%%%%%%%


\begin{thebibliography} {1}

\bibitem{Ahmad}
\newblock F. Ahmad, J. Cisewski, J. Xu, R. Anderson, 
\newblock Provisional mortality Data- United States, 2022,
\newblock CDC 72(18: 488-492, (2023).

\bibitem{NCJ2018}
\newblock M. Alper, M.R.Durosse, J. Markman, Update on Prisoner Recidivism, NCJ 250975. (2018).

\bibitem{BJSDead}
\newblock A. Carson.
\newblock Mortality in State and Federal Prisons 2001–2019, \newblock U.S. Department of Justice, Office of Justice Programs Bureau of Justice Statistics, (2023).

\bibitem{FBI_report}
\newblock FBI Report 2018
\newblock Crime in the Unite States
\newblock https://ucr.fbi.gov/crime-in-the-u.s/2018/crime-in-the-u.s.-2018/topic-pages/nations-two-crime-measures. 


\bibitem{McMillon}
\newblock McMillon, D., Simon, C.P., Morenoff, J.
\newblock Modeling the Underlying Dynamics of the Spread of Crime,
\newblock PLOS ONE, www.plosone.org 1  V. 9 I. 4 (1-22), (2014).

\bibitem{Martin}
\newblock J. Martin, M.P.H, Brady, E. Hamilton, and Michelle J.K. Osterman.
\newblock Births in the United States, CDC Data Brief,  No. 477 (2023).

\bibitem{VanWat}
\newblock Van den Driessche, P.,Watmough, J.,
\newblock Reproduction numbers and sub-threshold endemic equilibria for compartmental models of disease transmission,
\newblock Math. Biosci. 180; 2,
\newblock 9-48, (2002).

\bibitem{Lehrer}
\newblock Lehrer, D.,
\newblock Trauma-Informed Care: The Importance of Understanding the Incarcerated Women,
\newblock doi: 10.1089/jchc.20.07.0060, PMCID: PMC9041395, PMID: 34232778 (2021)

\bibitem {CCC_ZFeng}
\newblock C. Castillo-Chavez, Z. Feng, W. Huang. 
\newblock On the computation of R0 and its role on global stability. 
\newblock Springer, New York, pp 229–250. (2002).

\bibitem {Banks_CCC}
\newblock H.T. Banks, C. Castillo-Chavez, 
\newblock Bioterrorism: mathematical modeling applications in homeland security.
\newblock Society for Industrial and Applied Mathematics, Philadelphia, PA, (2003).

\bibitem {Clinard}
\newblock M.B. Clinard, R.F. Meier. Sociology of deviant behaviour. 
\newblock Wadsworth Publishing, Belmont. (2010).

\bibitem {Lawson_Marion}
\newblock Daniel Lawson and Glenn Marion. 
\newblock An introduction to mathematical modeling Glenn Marion, \newblock Bioinformatics and Statistics Scotland. (2008).

\bibitem {Diekman}
\newblock O. Diekmann, J.A.P. Heesterbeek, J.A. Metz, 
\newblock On the definition and the computation of the basic reproduction ratio $R_0$ in models for infectious diseases in heterogeneous populations,
\newblock J. Math. Biol. 28(4), 365-382, (1990).

\bibitem {Brauer}
\newblock F. Brauer, P. van den Driessche,
\newblock Models for transmission of disease with immigration of infectives, 
\newblock Math. Biosci. 171  143-154, (1995).

\bibitem {Becker}
\newblock G. Becker, 
\newblock Crime and punishment: an economic approach, J. Politic. Econ. 76, 169-217, (1968).

\bibitem {Glueck}
\newblock S. Glueck, E.T. Glueck, 
\newblock Unravelling Juvenile Delinquency, 
\newblock Harvard University Press, Cambridge, (1950).

\bibitem {Gordon}
\newblock M.B. Gordon, 
\newblock A random walk in the literature on criminality: a partial and critical view on some statistical
analyses and modelling approaches,
\newblock Eur. J. Appl. Math. 21, 283-306, (2010).

\bibitem {Hethcote} 
\newblock H.W. Hethcote. 
\newblock The mathematics of infectious diseases, 
\newblock SIAM Rev. 42(4), 599-653, (2000).

\bibitem {Zhao_Feng_CCC}
\newblock H. Zhao, Z. Feng, C. Castillo-Chavez,
\newblock The dynamics of poverty and crime.
\newblock J. Shanghai Normal Univ., Nat. Sci. Math. 43(5) , 486-495, (2014).

\bibitem{Ehrlich}
\newblock I. Ehrlich,
\newblock On the usefulness of controlling individuals: an economic analysis of rehabilitation, incapacitation
and deterrence,
\newblock Amer. Econ. Rev. 71(3), 307-322, (1981).

\bibitem {Crane}
\newblock J. Crane,
\newblock The epidemic theory of ghettos and neighborhood effects on dropping out and teenage childbearing,
\newblock Amer. J. Sociol. 96, 1226-1259. (1991).

\bibitem {Kretzschmar}
\newblock M. Kretzschmar, J. Wallinga,
\newblock Mathematical Models in Infectious Disease Epidemiology, 
\newblock Modern Infectious Disease Epidemiology, Springer New York, New York, NY, pp. 209–221, (2009).

\bibitem {Santoja}
\newblock F.J. Santonja, A.C. Tarazona, R.J. Villanueva,
\newblock A mathematical model of the pressure of an extreme ideology on a society,
\newblock Comput. Math. Appl. 56, 836–846, (2008).

\bibitem {Akers}
\newblock R.L. Akers, C.S. Sellers, 
\newblock Criminological theories: Introduction, evaluation, and application,
\newblock Rox-bury, Los Angeles, 5th Ed., (2008).


\end{thebibliography}
\end{document}